\documentclass[english]{article}
\usepackage[T1]{fontenc}
\usepackage[utf8]{inputenc}
\usepackage{geometry}
\usepackage{refstyle}
\usepackage{float}
\usepackage{amsthm}
\usepackage{amsmath}
\usepackage{setspace}

\makeatletter

\AtBeginDocument{\providecommand\corref[1]{\ref{cor:#1}}}
\RS@ifundefined{subref}
  {\def\RSsubtxt{section~}\newref{sub}{name = \RSsubtxt}}
  {}
\RS@ifundefined{thmref}
  {\def\RSthmtxt{theorem~}\newref{thm}{name = \RSthmtxt}}
  {}
\RS@ifundefined{lemref}
  {\def\RSlemtxt{lemma~}\newref{lem}{name = \RSlemtxt}}
  {}

\theoremstyle{plain}
\newtheorem{thm}{\protect\theoremname}[section]
  \theoremstyle{definition}
  \newtheorem{defn}[thm]{\protect\definitionname}
  \theoremstyle{plain}
  \newtheorem{lem}[thm]{\protect\lemmaname}
  \theoremstyle{plain}
  \newtheorem{cor}[thm]{\protect\corollaryname}

\usepackage{tikz}
\usepackage{pgfplots}
\usepackage{pgf}
\usepackage{multirow}
\usepackage{slashbox}
\usepackage{colortbl}

\usepackage[auth-lg]{authblk}

   \author[1]{Rotem Arnon Friedman}
   \author[2]{Esther H{\"a}nggi}
   \author[1]{Amnon Ta-Shma}
   \affil[1]{The Blavatnik School of Computer Science, Tel-Aviv University, Israel}
   \affil[2]{Centre for Quantum Technologies, National University of Singapore, Singapore}

\makeatother

\usepackage{babel}
  \providecommand{\corollaryname}{Corollary}
  \providecommand{\definitionname}{Definition}
  \providecommand{\lemmaname}{Lemma}
\providecommand{\theoremname}{Theorem}

\begin{document}

\title{Towards the Impossibility of Non-Signalling Privacy Amplification
from Time-Like Ordering Constraints}
\maketitle
\begin{abstract}
In the past few years there was a growing interest in proving the
security of cryptographic protocols, such as key distribution protocols,
from the sole assumption that the systems of Alice and Bob cannot
signal to each other. This can be achieved by making sure that Alice
and Bob perform their measurements in a space-like separated way (and
therefore signalling is impossible according to the non-signalling
postulate of relativity theory) or even by shielding their apparatus.
Unfortunately, it was proven in \cite{hanggi2010impossibility} that,
no matter what hash function we use, privacy amplification is impossible
if we only impose non-signalling conditions between Alice and Bob
and not within their systems. 

In this letter we reduce the gap between the assumptions of \cite{hanggi2010impossibility}
and the physical relevant assumptions, from an experimental point
of view, which say that the systems can only signal forward in time
within the systems of Alice and Bob. We consider a set of assumptions
which is very close to the conditions above and prove that the impossibility
result of \cite{hanggi2010impossibility} still holds.
\end{abstract}

\section{Introduction and Contribution }

\subsection{Non-signalling cryptography}

In the past few years there was a growing interest in proving the
security of cryptographic protocols, such as quantum key distribution
(QKD) protocols, from the sole assumption that the system on which
the protocol is being executed does not allow for signalling between
Alice and Bob. One way to make sure that this assumption holds is
for Alice and Bob to have secured shielded laboratories, such that
information cannot leak outside. It could also be ensured by performing
Alice's and Bob's measurements in a space-like separated way; this
way, relativity theory predicts the impossibility of signalling between
them. For this reason, such cryptographic protocols are sometimes
called {}``relativistic protocols''. Since the condition that information
cannot leak outside is a necessary condition in any cryptographic
protocol (otherwise the key could just leak out to the adversary,
Eve), basing the security proof on this condition alone will mean
that the protocol has minimal assumptions. 

We consider families of protocols which have two special properties.
First, the security of the protocols is based only on the observed
correlations of Alice's and Bob's measurements outcomes and not on
the physical apparatus they use. I.e., the protocols are device-independent
\cite{mayers1998quantum,pironio2009device}. In device-independent
protocols, we assume that the system of Alice and Bob was prepared
by the adversary Eve. Note that although the system was created by
Eve, Alice and Bob have to be able to make sure that information does
not leak outside by shielding the systems. Alice and Bob therefore
perform some (unknown) measurements on their system and privacy should
be concluded only from the correlations of the outcomes. 

Second, in the protocols that we consider, the adversary is limited
only by the non-signalling principle and not by quantum physics (i.e.,
super-quantum adversary). By combining these two properties together
we can say that quantum physics guarantees the protocol to work, but
the security is completely independent of quantum physics.

\subsection{Systems and correlations}

For two correlated random variables $X,U$ over $\varLambda_{1}\times\varLambda_{2}$,
we denote the conditional probability distribution of $X$ given $U$
by $P_{X|U}(x|u)=Pr(X=x|U=u)$.

A bipartite system is defined by the joint input-output behavior $P_{XY|UV}$
(see Figure \ref{fig:A-two-partite-system}). 

\begin{figure}[H]
\begin{centering}
\begin{tikzpicture}[>=stealth, thick, font=\large]
	\draw[->] (-1.5,1.5) node[anchor=east] {U} --(0,1.5) ;
	\draw[<-] (-1.5,0.5) node[anchor=east] {X} --(0,0.5);
	\draw (0,0) rectangle (3,2); 
	\draw (1.5,1) node {$P_{XY|UV}$};
	\draw[<-] (3,1.5)--(4.5,1.5) node[anchor=west] {V};
	\draw[->] (3,0.5)--(4.5,0.5) node[anchor=west] {Y};
\end{tikzpicture}
\par\end{centering}

\centering{}\caption{\label{fig:A-two-partite-system}A bipartite system}
\end{figure}
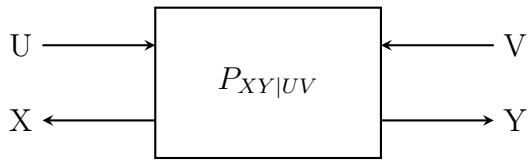

In a system $P_{XY|UV}$ $U$ and $X$ are usually Alice's input and
output respectively, while $V$ and $Y$ are Bob's input and output.
We denote Alice's interface of the system by $X(U)$ and Bob's interface
by $Y(V)$. In a similar way, when considering a tripartite system
$P_{XYZ|UVW}$ Eve's interface of the system is denoted by $Z(W)$.

We are interested in non-local systems - systems which cannot be described
by shared randomness of the parties. Bell proved in \cite{bell64}
that entangled quantum states can display non-local correlations under
measurements. Bell's work was an answer to Einstein, Podolsky, and
Rosen's claim in \cite{EPR} that quantum physics is incomplete and
should be augmented by classical variables determining the behavior
of every system under any possible measurement. In this letter we
deal with a specific type of Bell inequality, called the CHSH inequality
after \cite{CHSH}. 

We can think about the CHSH inequity as a game. In the CHSH game Alice
and Bob share a bipartite system $P_{XY|UV}$. Alice gets a random
input $U$, Bob gets a random input $V$ and the goal is that the
outputs of Alice and Bob, $X$ and $Y$ respectively, will satisfy
$X\oplus Y=U\text{·}V$. For all local systems the probability of
winning the game satisfies $\Pr[X\oplus Y=U\text{·}V]\leq0.75$.
This can be easily seen from the fact that only three out of the four
conditions represented by $\Pr[X\oplus Y=U\cdot V]=1$ can be satisfi{}ed
together. If a system violates the inequality then it is non-local. 
\begin{defn}
(CHSH non-locality). A system $P_{XY|UV}$ is non-local if $\underset{u,v}{\sum}\frac{1}{4}\Pr[X\oplus Y=u\cdot v]>0.75$.
\end{defn}
When measuring entangled quantum states, one can achieve roughly 85\%;
this is a Bell inequality violation. The maximal violation of the
CHSH inequality, i.e. $\underset{u,v}{\sum}\frac{1}{4}\Pr[X\oplus Y=u\cdot v]=1$
for any $u,v$, is achieved by the following system, called a Popescu-Rohrlich
box, or a PR-box \cite{PR-box}.
\begin{defn}
(PR-box). A PR box is the following bipartite system $P_{XY|UV}$:
For each input pair $(u,v)$, the random variables $X$ and $Y$ are
uniform bits and we have $\underset{u,v}{\sum}\frac{1}{4}\Pr[X\oplus Y=u\cdot v]=1$
(see Figure \ref{fig:PR-box}).

\begin{figure}
\begin{centering}
\begin{tikzpicture}[scale=0.5, font=\large]
	\draw[step=2] (-5,-4) grid (4,5);
	\draw[ultra thick] (-6,4)--(4,4);
	\draw[ultra thick] (-6,-4)--(4,-4);
	\draw[ultra thick] (-4,-4)--(-4,6);
	\draw[ultra thick] (4,-4)--(4,6);
	\draw[ultra thick] (-6,0)--(4,0);
	\draw[ultra thick] (0,-4)--(0,6);
	\draw (-4,4)--(-6,6);

	\draw (-3,3) node {$\frac{1}{2}$};
	\draw (-1,3) node {0};
	\draw (1,3) node {$\frac{1}{2}$};
	\draw (3,3) node {0};

	\draw (-3,1) node {0};
	\draw (-1,1) node {$\frac{1}{2}$};
	\draw (1,1) node {0};
	\draw (3,1) node {$\frac{1}{2}$};

	\draw (-3,-1) node {$\frac{1}{2}$};
	\draw (-1,-1) node {0};
	\draw (1,-1) node {0};
	\draw (3,-1) node {$\frac{1}{2}$};

	\draw (-3,-3) node {0};
	\draw (-1,-3) node {$\frac{1}{2}$};
	\draw (1,-3) node {$\frac{1}{2}$};
	\draw (3,-3) node {0};

	\draw (-5,3) node {0};
	\draw (-5,1) node {1};
	\draw (-5,-1) node {0};
	\draw (-5,-3) node {1};
	\draw (-6,2) node {0};
	\draw (-6,-2) node {1};

	\draw (-3,5) node {0};
	\draw (-1,5) node {1};
	\draw (1,5) node {0};
	\draw (3,5) node {1};
	\draw (-2,6) node {0};
	\draw (2,6) node {1};

	\draw (-5,4.4) node {Y};
	\draw (-4.4,5) node {X};
	\draw (-6,5) node {V};
	\draw (-5,6) node {U};
	  
\end{tikzpicture}
\par\end{centering}

\caption{\label{fig:PR-box}PR-box}

\end{figure}

\end{defn}
As seen from Figure \ref{fig:PR-box} the outputs are perfectly random,
and since the correlations are non-local, they cannot be described
by pre-shared randomness. I.e., PR-boxes correspond to perfect secrecy.
This implies that PR-boxes could have been a good resource for cryptographic
protocols. Unfortunately, perfect PR-boxes do not exist in nature;
as was proven by Tsirelson \cite{Tsirelson}, quantum physics is non-local,
but not maximally. Therefore, for a protocol which can be implemented
using quantum systems, we should consider approximations of PR-boxes,
or PR-boxes with some error. For example, an 85\%-approximations can
be achieved with maximally entangled qubits. For a more general treatment
we can define the following.
\begin{defn}
\label{PR-box-error}(Unbiased PR-box with error $\varepsilon$).
An unbiased PR-box with error $\varepsilon$ is the following bipartite
system $P_{XY|UV}$: For each input pair $(u,v)$, the random variables
$X$ and $Y$ are uniform bits and we have $\Pr[X\oplus Y=u\cdot v]=1-\varepsilon$
(see Figure \ref{fig:PR-box-error}).
\end{defn}
Note that the error here is the same error for all inputs. In a similar
way we can define different errors for different inputs.

\begin{figure}
\begin{centering}
\begin{tikzpicture}[scale=0.5, font=\large]
	\draw[step=2] (-5,-4) grid (4,5);
	\draw[ultra thick] (-6,4)--(4,4);
	\draw[ultra thick] (-6,-4)--(4,-4);
	\draw[ultra thick] (-4,-4)--(-4,6);
	\draw[ultra thick] (4,-4)--(4,6);
	\draw[ultra thick] (-6,0)--(4,0);
	\draw[ultra thick] (0,-4)--(0,6);
	\draw (-4,4)--(-6,6);

	\draw (-3,3) node {$\frac{1}{2}-\frac{\varepsilon}{2}$};
	\draw (-1,3) node {$\frac{\varepsilon}{2}$};
	\draw (1,3) node {$\frac{1}{2}-\frac{\varepsilon}{2}$};
	\draw (3,3) node {$\frac{\varepsilon}{2}$};

	\draw (-3,1) node {$\frac{\varepsilon}{2}$};
	\draw (-1,1) node {$\frac{1}{2}-\frac{\varepsilon}{2}$};
	\draw (1,1) node {$\frac{\varepsilon}{2}$};
	\draw (3,1) node {$\frac{1}{2}-\frac{\varepsilon}{2}$};

	\draw (-3,-1) node {$\frac{1}{2}-\frac{\varepsilon}{2}$};
	\draw (-1,-1) node {$\frac{\varepsilon}{2}$};
	\draw (1,-1) node {$\frac{\varepsilon}{2}$};
	\draw (3,-1) node {$\frac{1}{2}-\frac{\varepsilon}{2}$};

	\draw (-3,-3) node {$\frac{\varepsilon}{2}$};
	\draw (-1,-3) node {$\frac{1}{2}-\frac{\varepsilon}{2}$};
	\draw (1,-3) node {$\frac{1}{2}-\frac{\varepsilon}{2}$};
	\draw (3,-3) node {$\frac{\varepsilon}{2}$};

	\draw (-5,3) node {0};
	\draw (-5,1) node {1};
	\draw (-5,-1) node {0};
	\draw (-5,-3) node {1};
	\draw (-6,2) node {0};
	\draw (-6,-2) node {1};

	\draw (-3,5) node {0};
	\draw (-1,5) node {1};
	\draw (1,5) node {0};
	\draw (3,5) node {1};
	\draw (-2,6) node {0};
	\draw (2,6) node {1};

	\draw (-5,4.4) node {Y};
	\draw (-4.4,5) node {X};
	\draw (-6,5) node {V};
	\draw (-5,6) node {U};
	  
\end{tikzpicture}
\par\end{centering}

\caption{\label{fig:PR-box-error}Unbiased PR-box with error $\varepsilon$}
\end{figure}

Using this notation, systems $P_{XY|UV}$ which approximate the PR-box
with error $\varepsilon\in[0,0.25)$ are non-local. For a proof that
any unbiased PR-box with error $\varepsilon<0.25$ {}``holds'' some
secrecy, see for example Lemma 5 in \cite{hanggi2009quantum}. While
PR-Boxes correspond to perfect secrecy, PR-boxes with error correspond
to partial secrecy. The problem is that the amount of secrecy (defined
formally in Section \ref{sub:Distance-measures}) which can be achieved
from a quantum system is not enough for our purposes. Therefore we
must have some privacy amplification protocol in order for such systems
to be useful.

\subsection{Privacy amplification}

In the privacy amplification problem we consider the following scenario.
Alice and Bob share information that is only partially secret with
respect to an adversary Eve. Their goal is to distill this information
to a shorter string, the key, that is completely secret. The problem
was introduced in \cite{bennett1985privacy,bennett1995generalized}
for classical adversaries and in \cite{konigQuantumPA} for quantum
adversaries. In our case, Alice and Bob want to create a secret key
using a system $P_{XY|UV}$ while Eve, who is only limited by the
non-signalling principle, tries to get some information about it. 

Assume that Alice and Bob share a system from which they can create
a partially secret bit string $X$. Information theoretically, if
there is some entropy in one system, we can hope that by using several
systems we will have enough entropy to create a more secure key. The
idea behind privacy amplification is to consider Alice's and Bob's
system as a black box, take several such systems which will produce
several partially secret bit strings $X_{1},...,X_{n}$ and then apply
some hash function $f$ (which might take a short random seed as an
additional input) to $X_{1},...,X_{n}$, in order to receive a shorter
but more secret bit string $K$, which will act as the key. 

The amount of secrecy, as will be defined formally in Section \ref{sub:Distance-measures},
is usually measured by the distance of the actual system of Alice,
Bob and Eve from an ideal system, in which the key is uniformly distributed
and not correlated to the information held by Eve. We will denote
this distance by $d(K|E)$, where $E$ is Eve's system. We say that
a system generating a key is $\epsilon$-indistinguishable from an
ideal system if $d(K|E)\leq\epsilon$ for some small $\epsilon>0$.
Therefore, the problem of privacy amplification is actually the problem
of finding such a `good' function $f$. 

Privacy amplification is said to be possible when $\epsilon$ is a
decreasing function of $n$, the number of systems held by Alice and
Bob. In order to prove an impossibility result it is enough to give
a specific system, in which each of the subsystems holds some secrecy,
but this secrecy cannot be amplified by using any hash function -
the distance from uniform remains high, no matter what function Alice
and Bob apply to their output bits and how many systems they share. 

In the classical scenario, this problem can be solved almost optimally
by extractors \cite{nisan1999extracting,shaltiel2011introduction}.
Although not all classical extractors work against quantum adversaries
\cite{gavinsky2006exponential}, some very good extractors do, for
example, \cite{de2009trevisan}. Since we consider a super-quantum
adversary, we cannot assume that protocols which work for the classical
and quantum case, will stay secure against a more powerful adversary.
Therefore a different treatment is needed when considering non-signalling
adversaries.

\subsection{Related work }

Barrett, Hardy, and Kent have shown in \cite{barrett2005no} a protocol
for QKD which is based only on the assumption that Alice and Bob cannot
signal to each other. Unfortunately, the suggested protocol cannot
tolerate any errors caused by noise in the quantum channel and is
inefficient in the number of quantum systems used in order to produce
one secure bit. This problem could have been solved by using a privacy
amplification protocol, which works even when the adversary is limited
only by the non-signalling principle. Unfortunately, it was proven
in \cite{hanggi2010impossibility} that such a privacy amplification
protocol does not exist if signalling is possible within the laboratories
of Alice and Bob. 

On the contrary, in \cite{hanggi2009quantum}, \cite{masanes2009universally}
and \cite{masanes2011secure} it was proven that if we assume full
non-signalling conditions, i.e., that any subset of systems cannot
signal to any other subset of systems, QKD which is based only on
the non-locality of the correlations is possible. In particular, the
step of privacy amplification is possible. 

In the gap between these two extreme cases little has been known.
There is one particular set of assumptions of special interest from
an experimental point of view; the set of assumptions which says that
the systems can only signal forward in time within the systems of
Alice and Bob. For this setting it was only known that privacy amplification
using the XOR or the AND function is impossible \cite{MasanesXORimpossibility}.

\subsection{Contribution}

In this letter we reduce the gap between the assumptions of \cite{hanggi2010impossibility},
in which signalling is impossible only between Alice and Bob, and
the physical relevant assumptions which says that the systems can
only signal forward in time within the systems of Alice and Bob. We
consider a set of assumptions which is very close to the conditions
which only allow to signal forward in time and prove that the impossibility
result of \cite{hanggi2010impossibility} still holds. 

Since our set of assumptions differs only a bit from the assumptions
of signalling only forward in time, called {}``backward non-signalling'',
we can highlight the specific assumptions which might make the difference
between possibility and impossibility results. If the adversary does
not necessarily need to exploit these specific assumptions, then privacy
amplification will be impossible also in the physical assumptions
of {}``backward non-signalling'' systems. On the other hand, if
privacy amplification will be proved to be possible we will know that
the power of the adversary arises from these assumptions. 

The proof given here is an extension of the proof in \cite{hanggi2010impossibility}.
We prove that the adversarial strategy suggested in \cite{hanggi2010impossibility}
is still valid under stricter non-signalling assumptions (Theorem
\ref{thm:main}), and as a consequence also under the assumption of
an {}``almost backward non-signalling'' system (Corollary \ref{cor:seq-cor}).
This will imply that privacy amplification against non-signalling
adversaries is impossible under our stricter assumptions (which include
a lot more non-signalling conditions than in \cite{hanggi2010impossibility}),
as stated formally in Theorem \ref{thm:main}.

\subsection{Outline }

The rest of this letter is organized as follows. In Section \ref{sec:Preliminaries}
we describe several different non-signalling conditions and explain
the model of non-signalling adversaries. In Section \ref{sec:The-Underlying-System}
we define a specific system which respects many non-signalling conditions
and yet we cannot use privacy amplification in order to create an
arbitrary secure bit from it. In addition, we prove that an impossibility
result for our set of assumptions implies an impossibility result
for {}``almost backward non-signalling'' systems (Corollary \corref{seq-cor}).
In Section \ref{sec:Privacy-Amplification-Against} we prove our main
theorem, Theorem \ref{thm:main}. We conclude in Section \ref{sec:Concluding-Remarks}.

\section{Preliminaries \label{sec:Preliminaries}}

\subsection{Notations}

We denote the set $\{1,...,n\}$ by $[n]$. For any string $x\in\{0,1\}^{n}$
and any subset $I\subseteq[n]$, $x_{i}$ stands for the i'th bit
of $x$ and $x_{I}\in\{0,1\}^{|I|}$ stands for the string formed
by the bits of $x$ at the positions given by the elements of $I$.
$\overline{I}$ is the complementary set of $I$, i.e., $\overline{I}=[n]/I$.
$x_{\overline{i}}$ is the string formed by all the bits of $x$ except
for the i'th bit. 

For two correlated random variables $X,U$ over $\varLambda_{1}\times\varLambda_{2}$,
we denote the conditional probability distribution of $X$ given $U$
as $P_{X|U}(x|u)=Pr(X=x|U=u)$.

\subsection{Non-signalling systems\label{sub:Non-signaling-systems} \label{sub:Different-non-signaling}}

We start by formally defining the different types of non-signalling
systems and conditions which will be relevant in this letter. 
\begin{defn}
\noindent \label{n.s.-def}(Fully non-signalling system). An n-party
conditional probability distribution $P_{X|U}$ over $X,U\in\{0,1\}^{n}$
is called a fully non-signalling system if for any set $I\subseteq[n]$,
\[
\forall x_{\overline{I}},u_{I},u'_{I},u_{\overline{I}}\underset{x_{I}\in\{0,1\}^{|I|}}{\sum}P_{X|U}(x_{I},x_{\overline{I}}|u_{I},u_{\overline{I}})=\underset{x_{I}\in\{0,1\}^{|I|}}{\sum}P_{X|U}(x_{I},x_{\overline{I}}|u'_{I},u_{\overline{I}})\,.
\]

\end{defn}
This definition implies that any group of parties cannot infer from
their part of the system which inputs were given by the other parties.
A measurement of a subset $I$ of the parties does not change the
statistics of the outcomes of parties $\overline{I}$; the marginal
system they see is the same for all inputs of the other parties. This
means that different parties cannot signal to other parties using
only the system. Note that this type of condition is not symmetric.
The fact that parties $I$ cannot signal to parties $\overline{I}$
does not imply that parties $\overline{I}$ cannot signal to parties
$I$. The fully non-signalling conditions can also be written in the
following way. 
\begin{lem}
\noindent \label{lem:n.s.-equiv-def}(Lemma 2.7 in \cite{hanggi2010device}).
An n-party system $P_{X|U}$ over $X,U\in\{0,1\}^{n}$ is a fully
non-signalling system if and only if for all $i\in[n]$, 
\[
\forall x_{\overline{i}},u_{i},u'_{i},u_{\overline{i}}\underset{x_{i}\in\{0,1\}}{\sum}P_{X|U}(x_{i},x_{\overline{i}}|u_{i},u_{\overline{i}})=\underset{x_{i}\in\{0,1\}}{\sum}P_{X|U}(x_{i},x_{\overline{i}}|u'_{i},u_{\overline{i}})\,.
\]

\end{lem}
In order to make sure that the fully non-signalling conditions as
in Definition \ref{n.s.-def} hold one will have to create the system
such that each of the $2n$ subsystems is space-like separated from
all the others, or shielded, to exclude signalling. This is of course
impractical from an experimental point of view. Therefore, we need
to consider more practical, weaker, conditions. A minimal requirement
needed for any useful system is that Alice cannot signal to Bob and
vice versa%
\footnote{If we will not ensure this condition, say by making sure that they
are in space-like separated regions or by shielding their systems,
the measured Bell violation will have no meaning and any protocol
based on some kind of non locality will fail%
}. We stress that this is an assumption, since the non-signalling condition
cannot be tested (not even with some small error) using a parameter
estimation protocol as a previous step. This assumption can be justified
physically by shielding the systems or by performing the measurements
in a space-like separated way.
\begin{defn}
\label{Alice-&-Bob-n.s.}(Non-signalling between Alice and Bob). A
$2n$-party conditional probability distribution $P_{XY|UV}$ over
$X,Y,U,V\in\{0,1\}^{n}$ does not allow for signalling from Alice
to Bob if 
\[
\forall y,u,u',v\quad\underset{x}{\sum}P_{XY|UV}(x,y|u,v)=\underset{x}{\sum}P_{XY|UV}(x,y|u',v)
\]
and does not allow for signalling from Bob to Alice if 
\[
\forall x,v,v',u\quad\underset{y}{\sum}P_{XY|UV}(x,y|u,v)=\underset{y}{\sum}P_{XY|UV}(x,y|u,v')\,.
\]

\end{defn}
On top of the assumption that Alice and Bob cannot signal to each
other, we can now add different types of non-signalling conditions.
In a more mathematical way, we can think about it as follows. The
full non-signalling conditions are a set of linear equations as in
Definition \ref{n.s.-def} and Lemma \ref{lem:n.s.-equiv-def}. We
can assume that all of these equations hold (this is the full non-signalling
scenario) or we can use just a subset (which does not span the whole
set) of these equations. 

One type of systems which are physically interesting are the systems
which can only signal forward in time (messages cannot be sent to
the past). This can be easily achieved by measuring several quantum
systems one after another, and therefore these are the non-signalling
conditions that one {}``gets for free'' when performing an experiment
of QKD. For example, an entanglement-based protocol in which Alice
and Bob receive entangled photons and measure them one after another
using the same apparatus will lead to the conditions of Definition
\ref{alternative-seq-n.s.}. If the apparatus has memory signalling
is possible from $A_{i}$ to $A_{i+1}$ for example. However, signals
cannot go outside from Alice's laboratory to Bob's laboratory. Formally,
we use the following definition for backward non-signalling systems.
\begin{defn}
\label{alternative-seq-n.s.}(Backward non-signalling system). For
any $i\in\{2,...,n-1\}$ denote the set $\{1,...,i-1\}$ by $I_{1}$
and the set $\{i,...,n\}$ by $I_{2}$. A $2n$-party conditional
probability distribution $P_{XY|UV}$ over $X,Y,U,V\in\{0,1\}^{n}$
is a backward non-signalling system (does not allow for signalling
backward in time) if for any $i\in[n]$, 
\begin{eqnarray*}
\forall x_{I_{1}},y,u_{I_{1}},u_{I_{2}},u'_{I_{2}},v\quad\underset{x_{I_{2}}}{\sum}P_{XY|UV}(x_{I_{1}},x_{I_{2}},y|u_{I_{1}},u_{I_{2}},v) & = & \underset{x_{I_{2}}}{\sum}P_{XY|UV}(x_{I_{1}},x_{I_{2}},y|u_{I_{1}},u'_{I_{2}},v)\\
\forall x,y_{I_{1}},u,v_{I_{1}},v_{I_{2}},v'_{I_{2}}\quad\underset{y_{I_{2}}}{\sum}P_{XY|UV}(x,y_{I_{1}},y_{I_{2}}|u,v_{I_{1}},v_{I_{2}}) & = & \underset{y_{I_{2}}}{\sum}P_{XY|UV}(x,y_{I_{1}},y_{I_{2}}|u,v_{I_{1}},v'_{I_{2}}).
\end{eqnarray*}

\end{defn}
\begin{onehalfspace}
In order to understand why these are the conditions that we choose
to call {}``backward non-signalling'' note that in these conditions
Alice's (and analogously Bob's) systems $A_{I_{2}}$ cannot signal
not only to $A_{I_{1}}$, but even to $A_{I_{1}}$ and all of Bob's
systems together. I.e., $A_{I_{2}}$ cannot change the statistics
of $A_{I_{1}}$ and $B$, even if they are collaborating with one
another. Another way to see why these conditions make sense, is to
consider a scenario in which Bob, for example, performs all of his
measurements first. This of course should not change the results of
the experiment since Alice and Bob are separated and cannot send signals
to each other. Therefore when Alice performs her measurements on the
systems $A_{I_{2}}$, her outcomes cannot impact the statistics of
both $A_{I_{1}}$ and $B$ together. 
\end{onehalfspace}

In this letter we consider a different set of conditions, which does
not allow for most types of signalling to the past. 
\begin{defn}
\label{sequential-signaling}(Almost backward non-signalling system).
For any $i\in\{2,...,n-1\}$ denote the set $\{1,...,i-1\}$ by $I_{1}$
and the set $\{i,...,n\}$ by $I_{2}$. A $2n$-party conditional
probability distribution $P_{XY|UV}$ over $X,Y,U,V\in\{0,1\}^{n}$
is an almost backward non-signalling system if for any $i\in[n]$,
\begin{align*}
\forall x_{I_{1}},y_{I_{1}},u_{I_{1}},u_{I_{2}},u'_{I_{2}},v_{I_{1}},v_{I_{2}},v'_{I_{2}}\hphantom{------}\\
\underset{x_{I_{2}},y_{I_{2}}}{\sum}P_{XY|UV}(x_{I_{1}},x_{I_{2}},y_{I_{1}},y_{I_{2}}|u_{I_{1}},u_{I_{2}},v_{I_{1}},v_{I_{2}}) & =\underset{x_{I_{2}},y_{I_{2}}}{\sum}P_{XY|UV}(x_{I_{1}},x_{I_{2}},y_{I_{1}},y_{I_{2}}|u_{I_{1}},u'_{I_{2}},v_{I_{1}},v'_{I_{2}}).
\end{align*}

Figure \ref{fig:Different-non-signaling} visualizes the difference
between all of these non-signalling conditions. 
\end{defn}
\begin{figure}
\begin{tikzpicture} [font=\large]
	\begin{scope}
		\draw (0,0) node {$A_{n}$};
		\draw (0,0.8) node {.};
		\draw (0,1) node {.};
		\draw (0,1.2) node {.};
		\draw (0,2) node {$A_{3}$};
		\draw (0,3) node {$A_{2}$};
		\draw (0,4) node {$A_{1}$};

		\draw (2,0) node {$B_{n}$};
		\draw (2,0.8) node {.};
		\draw (2,1) node {.};
		\draw (2,1.2) node {.};
		\draw (2,2) node {$B_{3}$};
		\draw (2,3) node {$B_{2}$};
		\draw (2,4) node {$B_{1}$};

		\draw[color=red, dashed, thick] (1,-0.25) -- (1,4.25);
		\draw[<->, color=red, thick] (0.5,2) -- (1.5,2);
		\draw (1,-0.75) node[font=\small] {(a)};
	\end{scope}

	\begin{scope}[shift={(4,0)}]
		\draw (0,0) node {$A_{n}$};
		\draw (0,0.8) node {.};
		\draw (0,1) node {.};
		\draw (0,1.2) node {.};
		\draw (0,2) node {$A_{3}$};
		\draw (0,3) node {$A_{2}$};
		\draw (0,4) node {$A_{1}$};

		\draw (2,0) node {$B_{n}$};
		\draw (2,0.8) node {.};
		\draw (2,1) node {.};
		\draw (2,1.2) node {.};
		\draw (2,2) node {$B_{3}$};
		\draw (2,3) node {$B_{2}$};
		\draw (2,4) node {$B_{1}$};

		\draw[color=red, dashed, thick] (-0.25,2.5) -- (2.25,2.5);
		\draw[->, color=red, thick] (1.25,2.1) -- (1.25,2.9);
		\draw[<-, decorate, decoration={snake, amplitude=.4mm, segment length=2mm}, color=blue, thick] (0.75,2.1) -- (0.75,2.9);
		\draw (1,-0.75) node[font=\small] {(b)};
	\end{scope}
	
	\begin{scope}[shift={(8,0)}]
		\draw (0,0) node {$A_{n}$};
		\draw (0,0.8) node {.};
		\draw (0,1) node {.};
		\draw (0,1.2) node {.};
		\draw (0,2) node {$A_{3}$};
		\draw (0,3) node {$A_{2}$};
		\draw (0,4) node {$A_{1}$};

		\draw (2,0) node {$B_{n}$};
		\draw (2,0.8) node {.};
		\draw (2,1) node {.};
		\draw (2,1.2) node {.};
		\draw (2,2) node {$B_{3}$};
		\draw (2,3) node {$B_{2}$};
		\draw (2,4) node {$B_{1}$};

		\draw[color=red, dashed, thick] (-0.25,2.5) -- (0.5,2.5);
		\draw[color=red, dashed, thick] (0.5,2.5) -- (0.5,-0.25);
		\draw[->, color=red, thick] (0.25,1.5) -- (0.75,2);

		\draw[color=red, dashed, thick] (1.5,2.5) -- (2.25,2.5);
		\draw[color=red, dashed, thick] (1.5,2.5) -- (1.5,-0.25);
		\draw[->, color=red, thick] (1.75,1.5) -- (1.25,2);

		\draw (1,-0.75) node[font=\small] {(c)};
	\end{scope}

	\begin{scope}[shift={(12,0)}, color=red, thick, text=black]
		\draw (0,0) node[draw, dashed] {$A_{n}$};
		\draw[->] (0.25,0) -- (0.75,0.2);
		\draw (0,0.8) node {.};
		\draw (0,1) node {.};
		\draw (0,1.2) node {.};
		\draw (0,2) node[draw, dashed] {$A_{3}$};
		\draw[->] (0.25,2) -- (0.75,2);
		\draw (0,3) node[draw, dashed] {$A_{2}$};
		\draw[->] (0.25,3) -- (0.75,3);
		\draw (0,4) node[draw, dashed] {$A_{1}$};
		\draw[->] (0.25,4) -- (0.75,3.8);

		\draw (2,0) node[draw, dashed] {$B_{n}$};
		\draw[->] (1.75,0) -- (1.25,0.2);
		\draw (2,0.8) node {.};
		\draw (2,1) node {.};
		\draw (2,1.2) node {.};
		\draw (2,2) node[draw, dashed] {$B_{3}$};
		\draw[->] (1.75,2) -- (1.25,2);
		\draw (2,3) node[draw, dashed] {$B_{2}$};
		\draw[->] (1.75,3) -- (1.25,3);
		\draw (2,4) node[draw, dashed] {$B_{1}$};
		\draw[->] (1.75,4) -- (1.25,3.8);

		\draw (1,-0.75) node[font=\small] {(d)};
	\end{scope}

\end{tikzpicture}

\caption{\label{fig:Different-non-signaling}Different non-signalling conditions:
signalling is impossible in the direction of the straight red arrow.
(a) Non-signalling between Alice and Bob. (b) The conditions of Definition
\ref{sequential-signaling}, almost backward non-signalling conditions,
for $i=3$. Note that signalling may be possible in the direction
of the curly blue arrow. (c) The conditions of Definition \ref{alternative-seq-n.s.},
backward non-signalling conditions, for $i=3$. (d) Full non-signalling
conditions. The conditions we consider are the combination of (a)
and (b).}
\end{figure}
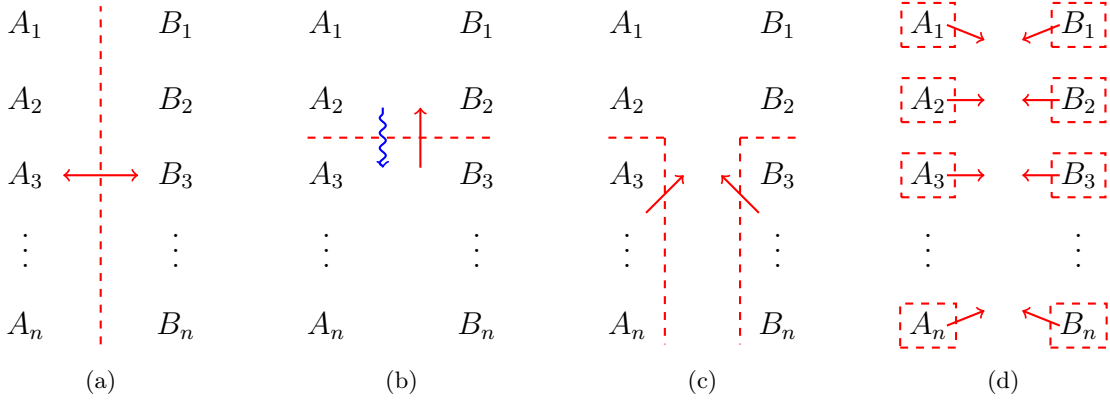

The difference between the conditions of Definition \ref{alternative-seq-n.s.}
and Definition \ref{sequential-signaling} is that when assuming the
conditions of an almost backward non-signalling system signalling
is not forbidden from $A_{i}$ to $B_{i}$ and $A_{j}$ together for
any $i$ and $j<i$. I.e., if $A_{i}$ wants to signal to some system
in the past, $A_{j}$, $B_{i}$ has to cooperate with $A_{j}$. To
see this consider the following system for example. Alice and Bob
share a system $P_{XY|UV}$ for $X,Y,U,V\in\{0,1\}^{2}$. We define
the system such that each of the outputs is a perfectly random bit
and independent of any input, except for $X_{1}$, which is equal
to $Y_{2}\oplus U_{2}$. Obviously, the outputs on Bob’s side look
completely random and independent of any input, i.e., the system is
non-signalling from Alice to Bob. Now note that whenever we do not
have access to $Y_{2}$, $X_{1}$ also looks like a perfectly random
bit and independent of the input. Therefore, the system is also non-signalling
from Bob to Alice, and almost backward non-signalling. However, the
conditions of Definition \ref{alternative-seq-n.s.} does not hold,
since the input $U_{2}$ can be perfectly known from $X_{1}$ and
$Y_{2}$ (i.e. $A_{2}$ can signal $A_{1}$ and $B_{2}$ together).

For every system $P_{XY|UV}$ which fulfills some arbitrary non-signalling
conditions we can define marginal systems and extensions to the system
in the following way. 
\begin{defn}
(Marginal system). A system $P_{X|U}$ is called a marginal system
of the system $P_{XZ|UW}$ if $\forall x,u,w\quad P_{X|U}(x|u)=\underset{z}{\sum}P_{XZ|UW}(x,z|u,w)$.
\end{defn}
Note that for the marginal system $P_{X|U}$ of $P_{XZ|UW}$ to be
defined properly, all we need is a non-signalling condition between
the parties holding $X(U)$ and the parties holding $Z(W)$. 
\begin{defn}
\noindent (Extension system). A system $P_{XZ|UW}$ is called an extension
to the system $P_{X|U}$, which fulfills some arbitrary set of non-signalling
conditions $\mathcal{C}$, if:
\begin{enumerate}
\item $P_{XZ|UW}$ does not allow for signalling between the parties holding
$X(U)$ and the parties holding $Z(W)$.
\item \noindent The marginal system of $P_{XZ|UW}$ is $P_{X|U}$.
\item For any $z$ the system $P_{X|U}^{Z=z}$ fulfills the same non-signalling
conditions $\mathcal{C}$.
\end{enumerate}
\end{defn}
Note that for every system $P_{X|U}$ there are many different extensions.
Next, in an analogous way to the definition of a classical-quantum
state, $\rho_{XE}=\underset{x}{\sum}P_{X}(x)|x\rangle\langle x|\otimes\rho_{E}^{x}$,
we would like to define a classical-non-signalling system. 
\begin{defn}
\noindent (Classical - non-signalling system). A classical - non-signalling
(c-n.s.) system is a system $P_{XZ|UW}$ such that $|U|=1$. 
\end{defn}
\noindent We can think about it as a system in which some of the parties
cannot choose or change the input on their side of the system. When
it is clear from the context which side of the system is classical
and which side is not we drop the index which indicates the trivial
choice for $U$ and just write $P_{XZ|W}$. Notice that for a general
system with some $U$, after choosing an input $u_{i}\in U$ we get
the c-n.s. system $P_{XZ|U=u_{i},W}$.

\subsection{Distance measures\label{sub:Distance-measures}}

In general, the distance between any two systems $P_{X|U}$ and $Q_{X|U}$
can be measured by introducing another entity - the distinguisher.
Suppose $P_{X|U}$ and $Q_{X|U}$ are two known systems. The distinguisher
gets one of these systems, $S$, and has to guess which system he
was given. In the case of our non-signalling systems, the distinguisher
can choose which measurements to make (which inputs to insert to the
system) and to see all the outputs. He then outputs a bit $B$ with
his guess. The distinguishing advantage between systems $P_{X|U}$
and $Q_{X|U}$ is the maximum guessing advantage the best distinguisher
can have.
\begin{defn}
(Distinguishing advantage). The distinguishing advantage between two
systems $P_{X|U}$ and $Q_{X|U}$ is 
\[
\delta(P_{X|U},Q_{X|U})=\underset{D}{max}[P(B=1|S=P_{X|U})-P(B=1|S=Q_{X|U})]
\]
where the maximum is over all distinguishers $D$, $S$ is the system
which is given to the distinguisher and $B$ is its output bit. Two
systems $P_{X|U}$ and $Q_{X|U}$ are called $\epsilon$-indistinguishable
if $\delta(P_{X|U},Q_{X|U})\leq\epsilon$.
\end{defn}
If the distinguisher is given an n-party system for $n>1$ he can
choose not only the n inputs but also the order in which he will insert
them. The distinguisher can be adaptive, i.e., after choosing an input
and seeing an output he can base his later decisions for the following
inputs on the results seen so far. Therefore the maximization in this
case will be on the order of the measurements and their values. 

If the distinguisher is asked to distinguish between two c-n.s. systems
we can equivalently write the distinguishing advantage as in the following
lemma.
\begin{lem}
\noindent \label{lem:distance c-n.s}(Distinguishing advantage between
two c-n.s. systems). The distinguishing advantage between two c-n.s
systems $P_{KZ|W}$ and $Q_{KZ|W}$ is 
\[
\text{\ensuremath{\delta}}(P_{KZ|W},Q_{KZ|W})=\underset{k}{\sum}\underset{w}{max}\underset{z}{\sum}\biggm|P_{KZ|W=w}(k,z)-Q_{KZ|W=w}(k,z)\biggm|.
\]
\end{lem}
\begin{proof}
In order to distinguish between two c-n.s. systems, $P_{KZ|W}$ and
$Q_{KZ|W}$, the distinguisher has only one input to choose, $W$.
In addition, because the distinguisher has no choice for the input
of the classical part, the distinguishing advantage can only increase
if the distinguisher will first read the classical part of the system
and then choose $W$ according to the value of $K$. Therefore, for
two c-n.s. systems, the best strategy will be to read $K$ and then
to choose the best $W$, as indicated in the expression above.
\end{proof}
The distance (in norm 1) between two systems is defined to be half
of the distinguishing advantage between these systems.
\begin{defn}
(Distance between two c-n.s. systems). The distance between two c-n.s
systems $P_{KZ|W}$ and $Q_{KZ|W}$ in norm 1 is 
\[
\biggm|P_{KZ|W}-Q_{KZ|W}\biggm|_{1}\equiv\frac{1}{2}\underset{k}{\sum}\underset{w}{max}\underset{z}{\sum}\biggm|P_{KZ|W=w}(k,z)-Q_{KZ|W=w}(k,z)\biggm|.
\]

\end{defn}
In a cryptographic setting, we mostly consider the distance between
the real system in which the key is being calculated from the output
of the system held by the parties, and an ideal system. The ideal
system in our case is a system in which the key is uniformly distributed
and independent of the adversary's system. For a c-n.s. system $P_{KZ|W}$
where $K$ is over $\{0,1\}^{n}$, let $U_{n}$ denote the uniform
distribution over $\{0,1\}^{n}$ and let $P_{Z|W}$ be the marginal
system held by the adversary. The distance from uniform is a defined
as follows.
\begin{defn}
\label{Distance-from-uniform}(Distance from uniform). The distance
from uniform of the c-n.s. system $P_{KZ|W}$ is 
\[
d(K|Z(W))\equiv\biggm|P_{KZ|W}-U_{n}\times P_{Z|W}\biggm|_{1}
\]
where the system $U_{n}\times P_{Z|W}$ is defined such that $U_{n}\times P_{Z|W}(k,z|w)=U_{n}(k)\cdot P_{Z|W}(z|w)$.
\end{defn}
In the following sections we consider the distance from uniform given
a specific input (measurement) of the adversary, $W=w$. In this case,
according to Definition \ref{Distance-from-uniform}, we get 
\begin{eqnarray}
d(K|Z(w)) & = & \frac{1}{2}\underset{k,z}{\sum}\biggm|P_{KZ|W=w}(k,z)-U_{n}(k)\cdot P_{Z|W=w}(z)\biggm|=\nonumber \\
 & = & \frac{1}{2}\underset{k,z}{\sum}P_{Z|W=w}(z)\biggm|P_{K|Z=z}(k)-\frac{1}{n}\biggm|.\label{eq:dist}
\end{eqnarray}

\subsection{Modeling non-signalling adversaries}

When modeling a non-signalling adversary, the question in mind is:
given a system $P_{XY|UV}$ shared by Alice and Bob, for which some
arbitrary non-signalling conditions hold, which extensions to a system
$P_{XYZ|UVW}$, including the adversary Eve, are possible? The only
principle which limits Eve is the non-signalling principle, which
means that the conditional system $P_{XY|UV}^{Z=z}$ , for any $z\in Z$,
must fulfill all of the non-signalling conditions that $P_{XY|UV}$
fulfills, and in addition $P_{XYZ|UVW}$ does not allow signalling
between Alice and Bob together and Eve. Since any non-signalling assumptions
about the system of Alice and Bob are ensured physically (by shielding
the systems for example), they must still hold even if Eve's output
$z$ is given to some other party. Therefore the conditional system
$P_{XY|UV}^{Z=z}$ must also fulfill all the non-signalling conditions
of $P_{XY|UV}$, which justifies our assumptions about the power of
the adversary in this setting. 

\begin{figure}[H]
\begin{centering}
\begin{tikzpicture}[>=stealth, thick, font=\large]
	\draw[->] (-1.5,1.5) node[anchor=east] {U} --(0,1.5) ;
	\draw[<-] (-1.5,0.5) node[anchor=east] {X} --(0,0.5);
	\draw (0,0) rectangle (3,2); 
	\draw (1.5,1) node {$P_{XYZ|UVW}$};
	\draw[<-] (3,1.5)--(4.5,1.5) node[anchor=west] {V};
	\draw[->] (3,0.5)--(4.5,0.5) node[anchor=west] {Y};
	\draw[<-] (1,-1) node[anchor=north] {Z}--(1,0);
	\draw[->] (2,-1) node[anchor=north] {W} --(2,0);
\end{tikzpicture}
\par\end{centering}

\centering{}\caption{\label{fig:A-three-partite-system}A three-partite system}
\end{figure}
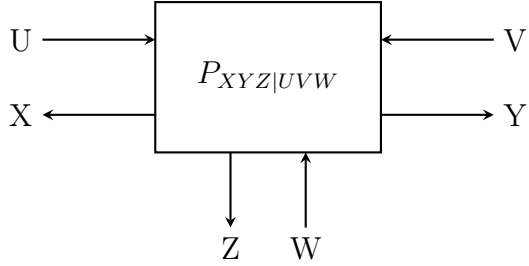

We adopt here the model given in \cite{hanggi2010impossibility,hanggi2010device,hanggi2009quantum}
of non-signalling adversaries. We reduce the scenario in which Alice,
Bob and Eve share a system $P_{XYZ|UVW}$ to the scenario considering
only Alice and Bob in the following way. Because Eve cannot signal
to Alice and Bob (even together) by her choice of input, we must have,
for all $x,y,u,v,w,w'$, 
\[
\sum_{z}P_{XYZ|UVW}(x,y,z|u,v,w)=\sum_{z}P_{XYZ|UVW}(x,y,z|u,v,w')=P_{XY|UV}(x,y|u,v).
\]
Moreover, as said before, since any non-signalling condition must
still hold even if Eve's output $z$ is given to some other party,
the system conditioned on Eve's outcome, $P_{XY|UV}^{Z=z}$, must
also fulfill all the non-signalling conditions of $P_{XY|UV}$. We
can therefore see Eve’s input as a choice of a convex decomposition
of Alice’s and Bob’s system and her output as indicating one part
of this decomposition. Formally,
\begin{defn}
(Partition of the system). A partition of a given multipartite system
$P_{XY|UV}$, which fulfills a certain set of non-signalling conditions
$\mathcal{C}$, is a family of pairs $(p^{z},P_{XY|UV}^{z})$, where:
\begin{enumerate}
\item $p^{z}$ is a classical distribution (i.e. for all $z$ $p^{z}\geq0$
and $\underset{z}{\sum}\, p^{z}=1$).
\item For all $z$, $P_{XY|UV}^{z}$ is a system that fulfills $\mathcal{C}$.
\item $P_{XY|UV}=\underset{z}{\sum}\, p^{z}\cdot P_{XY|UV\,.}^{z}$
\end{enumerate}
\end{defn}
We can use the same proof as in Lemma 2 and 3 in \cite{hanggi2009quantum}
to prove that this is indeed a legitimate model, i.e., that the set
of all partitions covers exactly all the possible strategies of a
non-signalling adversary in our case.

It is further proven in \cite{hanggi2010impossibility} that for showing
an impossibility result, we can assume that Eve's information $Z$
is a binary random variable: 
\begin{lem}
(Lemma 5 in \cite{hanggi2010impossibility}). If $(p^{z=0},P_{XY|UV}^{z=0})$
is an element of a partition with $m$ elements, then it is also possible
to define a new partition with only two elements, in which one of
the elements is $(p^{z=0},P_{XY|UV}^{z=0})$ .
\end{lem}
Moreover, it is not necessary to determine both parts of the partition
($(p^{z=0},P_{XY|UV}^{z=0})$ and\\
 $(p^{z=1},P_{XY|UV}^{z=1})$) explicitly. Instead, a condition on
the system given outcome $z=0$ is given, which will make sure that
there exists a second part, complementing it to a partition:
\begin{lem}
\label{lem:one-part-enough}(Lemma 6 in \cite{hanggi2010impossibility}).
Given a non-signalling distribution $P_{XY|UV}$, there exists a partition
with element $(p^{z=0},P_{XY|UV}^{z=0})$ if and only if for all inputs
and outputs $x,y,u,v$ it holds that $p^{z=0}\cdot P_{XY|UV}^{z=0}(x,y|u,v)\le P_{XY|UV}(x,y|u,v)$.
\end{lem}
For the formal proofs of these lemmas, note that since the non-signalling
conditions are linear the same proofs as in Lemma 5 and Lemma 6 in
\cite{hanggi2010impossibility} will hold here as well, no matter
which non-signalling conditions are imposed for $P_{XY|UV}$.

Defining a partition is equivalent to choosing a measurement $W=w$,
therefore, we can also write the distance from uniform of a key, as
in Equation (\ref{eq:dist}), using the partition itself. Since we
will only need to consider the case where Alice and Bob try to output
one secret bit, we can further simplify the expression, as in the
following lemma.
\begin{lem}
(Lemma 5.1 in \cite{hanggi2010device}). For the case $K=f(X)$, where
$f:\{0,1\}^{|X|}\rightarrow\{0,1\}$, $U=u$, $V=v$, and where the
strategy $W=w$ is defined by the partition $\left\{ (p^{z_{w}},P_{XY|UV}^{z_{w}})\right\} _{z_{w}\in\{0,1\}}$,
\[
d\left(K|Z(w)\right)=\frac{1}{2}\sum_{z_{w}}p^{z_{w}}\cdot\biggm|\sum_{x,y}(-1)^{f(x)}P_{XY|UV}^{z_{w}}(x,y|u,v)\biggm|.
\]

\end{lem}
For a proof see Lemma 5.1 in \cite{hanggi2010device}.

\section{The Non-signalling Assumptions \label{sec:The-Underlying-System}}

\subsection{The basic assumptions}

It was proven in \cite{hanggi2009quantum} (Lemma 5) that any unbiased
PR-box with error $\varepsilon<0.25$ holds some secrecy. With the
goal of amplifying the privacy of the secret in mind, Alice and Bob
now share $n$ such systems. The underlying system of Alice and Bob
that we consider is a product of $n$ independent PR-boxes with errors
(Definition \ref{PR-box-error}), as seen from Alice's and Bob's point
of view. This is stated formally in the following definition:
\begin{defn}
\label{A-product-system}\label{basic-system}(Product system). A
product system of $n$ copies of PR-boxes with error $\varepsilon$
is the system $P_{XY|UV}=\underset{i\in[n]}{\prod}P_{X_{i}Y_{i}|U_{i}V_{i}}$,
where for each $i$, the system $P_{X_{i}Y_{i}|U_{i}V_{i}}$ is an
unbiased PR-box with error $\varepsilon$ as in Definition \ref{PR-box-error}.
\end{defn}
In addition, as explained in Section \ref{sub:Different-non-signaling},
in order for any system to be useful, we will always make sure that
Alice and Bob cannot signal to each other (otherwise any non-local
violation will not have any meaning - it could have also been achieved
by signalling between the systems). Mathematically, this means that
for any outcome $z$ of any adversary, Alice and Bob cannot signal
to each other using the system $P_{XY|UV}^{z}$. I.e., $P_{XY|UV}^{z}$
fulfills the conditions of Definition \ref{Alice-&-Bob-n.s.}. 

On top of this assumption we can now add more non-signalling assumptions
of different types. For example, in \cite{hanggi2009quantum}, \cite{masanes2009universally}
and \cite{masanes2011secure} it was proven that if we assume full
non-signalling conditions then privacy amplification is possible.
On the contrary, in \cite{hanggi2010impossibility} it was proven
that if we do not add more non-signalling assumption (and use only
the assumption that Alice and Bob cannot signal to each other) then
privacy amplification is impossible. An interesting question is therefore,
what happens in the middle? Is privacy amplification possible when
we use some additional assumptions but not all of them? 

The goal of this letter is to consider the conditions of almost backward
non-signalling systems, given in Definition \ref{sequential-signaling}.
We will do so by considering a larger set of equations, defined formally
in Section \ref{sub:Our-assumptions}.

\subsection{Our additional assumptions \label{sub:Our-assumptions}}

Consider the following system. 
\begin{defn}
\label{our-system} Alice and Bob and Eve share a system $P_{XYZ|UVW}$
such that:
\begin{enumerate}
\item The marginal system of Alice and Bob $P_{XY|UV}$ is a product system
as in Definition \ref{basic-system}.
\item For any $z$, $P_{XY|UV}^{z}$ fulfills the conditions of Definition
\ref{Alice-&-Bob-n.s.} (Alice and Bob cannot signal each other).
\item For all $i\in[n]$ and for any $z$ 
\begin{eqnarray*}
\forall x_{\overline{i}},y_{\overline{i}},u_{i},u'_{i},u_{\overline{i}},v\qquad\underset{x_{i},y_{i}}{\sum}P_{XY|UV}^{z}(x,y|u,v) & = & \underset{x_{i},y_{i}}{\sum}P_{XY|UV}^{z}(x,y|u',v)\\
\forall x_{\overline{i}},y_{\overline{i}},u,v_{i},v'_{i},v_{\overline{i}}\qquad\underset{x_{i},y_{i}}{\sum}P_{XY|UV}^{z}(x,y|u,v) & = & \underset{x_{i},y_{i}}{\sum}P_{XY|UV}^{z}(x,y|u,v').
\end{eqnarray*}

\end{enumerate}
\end{defn}
Note that the set of these conditions is equivalent to 
\begin{equation}
\forall x_{\overline{i}},y_{\overline{i}},u_{i},u'_{i},u_{\overline{i}},v_{i},v'_{i},v_{\overline{i}}\qquad\underset{x_{i},y_{i}}{\sum}P_{XY|UV}^{z}(x,y|u,v)=\underset{x_{i},y_{i}}{\sum}P_{XY|UV}^{z}(x,y|u',v').\label{eq:equiv}
\end{equation}
To see this first note that the conditions of Definition \ref{our-system}
are a special case of Equation (\ref{eq:equiv}). For the second direction:
$\forall x_{\overline{i}},y_{\overline{i}},u_{i},u'_{i},u_{\overline{i}},v_{i},v'_{i},v_{\overline{i}}$,
\[
\underset{x_{i},y_{i}}{\sum}P_{XY|UV}^{z}(x,y|u,v)=\underset{x_{i},y_{i}}{\sum}P_{XY|UV}^{z}(x,y|u',v)=\underset{x_{i},y_{i}}{\sum}P_{XY|UV}^{z}(x,y|u',v').
\]
Therefore, the equations of Definition \ref{our-system} mean that
for all $i$, parties $A_{i}$ and $B_{i}$ together cannot signal
the other parties (See Figure \ref{fig:our-conditions}). 

\begin{figure}
\begin{centering}
\begin{tikzpicture} [font=\large]
	\draw (0,0) node {$A_{n}$};
	\draw (0,0.8) node {.};
	\draw (0,1) node {.};
	\draw (0,1.2) node {.};
	\draw (0,2) node {$A_{3}$};
	\draw (0,3) node {$A_{2}$};
	\draw (0,4) node {$A_{1}$};

	\draw (3,0) node {$B_{n}$};
	\draw (3,0.8) node {.};
	\draw (3,1) node {.};
	\draw (3,1.2) node {.};
	\draw (3,2) node {$B_{3}$};
	\draw (3,3) node {$B_{2}$};
	\draw (3,4) node {$B_{1}$};

	\draw[color=red, dashed, thick] (-0.5,1.5) rectangle (3.5,2.5);
	\draw[->, color=red, thick] (1.5,2.25) -- (1.5,2.75);
	\draw[<-, color=red, thick] (1.5,1.25) -- (1.5,1.75);
\end{tikzpicture}

\par\end{centering}

\caption{\label{fig:our-conditions}The n.s. conditions of Definition \ref{our-system}
for $i=3$}
\end{figure}
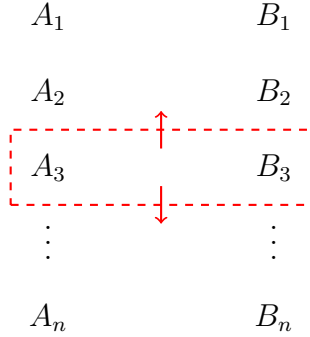

Adding these assumptions to the the non-signalling assumption between
Alice and Bob (Definition \ref{Alice-&-Bob-n.s.}) does not imply
the full non-signalling conditions. To see this consider the following
example. Alice and Bob share a system $P_{XY|UV}$ such that $X,Y,U,V\in\{0,1\}^{2}$.
We define the system such that each of the outputs is a perfectly
random bit and independent of any input, except for $X_{2}$, which
is equal to $Y_{1}\oplus U_{1}$. The outputs on Bob’s side look completely
random and independent of any input, i.e., the system is non-signalling
from Alice to Bob. Now note that whenever we do not have access to
$Y_{1}$, then $X_{2}$ also looks like a perfectly random bit and
independent of the input. Therefore, the system is also non-signalling
from Bob to Alice, and the conditions of Definition \ref{our-system}
hold as well. However, this system is not fully non-signalling, since
the input $U_{1}$ can be perfectly known from $X_{2}$ and $Y_{1}$
(i.e. $A_{1}$ can signal $A_{2}$ and $B_{1}$ together).

Adding this set of equations as assumptions means to add a lot more
assumptions about the system (on top of the basic system described
before). Intuitively, such a system is close to being a fully non-signalling
system. We will prove that even in this case, Theorem 15 in \cite{hanggi2010impossibility}
still holds and privacy amplification is impossible:
\begin{thm}
\label{thm:main}There exists a system as in Definition \ref{our-system}
such that for any hash function $f$, there exists a partition $w$
for which the distance from uniform of $f(X)$ given $w$ is at least
$c(\varepsilon)$, i.e., $d(f(X)|Z(w))\geq c(\varepsilon)$, where
$c(\varepsilon)$ is some constant which depends only on the error
of a single box, $\varepsilon$ (as in Definition \ref{A-product-system}).
\end{thm}
Note that although our set of equations might seem unusual, proving
an impossibility result for this set implies the same impossibility
result for all sets of linear equations that are determined by it.
The set of equations of an almost backward non-signalling system,
as in Definition \ref{sequential-signaling}, is one interesting example
of such a set. 
\begin{lem}
\label{lem:our-imply-sequentialy-signaling}The almost backward non-signalling
conditions, as in Definition \ref{sequential-signaling}, are implied
by the non-signalling conditions of Definition \ref{our-system}.\end{lem}
\begin{proof}
Consider the set of equations in Definition \ref{sequential-signaling}.
We will now prove them using the equations in Definition \ref{our-system},
this will imply that if the assumptions of Definition \ref{our-system}
hold then so do the assumption of almost backward non-signalling. 

For every $i\in[n]$ we can write 

\begin{multline*}
\underset{x_{I_{2}},y_{I_{2}}}{\sum}P_{XY|UV}(x,y|u_{I_{1}},u_{I_{2}},v_{I_{1}},v_{I_{2}})=\shoveright{\underset{\begin{array}{c}
x_{I_{2}/\{i\}}\\
y_{I_{2}/\{i\}}
\end{array}}{\sum}\underset{\begin{array}{c}
x_{i}\\
y_{i}
\end{array}}{\sum}P_{XY|UV}(x,y|u_{I_{1}},u_{i},u_{I_{2}/\{i\}},v_{I_{1}},v_{i},v_{I_{2}/\{i\}})}\\
=\shoveright{\underset{\begin{array}{c}
x_{I_{2}/\{i\}}\\
y_{I_{2}/\{i\}}
\end{array}}{\sum}\underset{\begin{array}{c}
x_{i}\\
y_{i}
\end{array}}{\sum}P_{XY|UV}(x,y|u_{I_{1}},u'_{i},u_{I_{2}/\{i\}},v_{I_{1}},v'_{i},v_{I_{2}/\{i\}})}\\
=\underset{\begin{array}{c}
x_{I_{2}/\{i+1\}}\\
y_{I_{2}/\{i+1\}}
\end{array}}{\sum}\underset{\begin{array}{c}
x_{i+1}\\
y_{i+1}
\end{array}}{\sum}P_{XY|UV}(x,y|u_{I_{1}},u'_{i},u_{i+1},u_{I_{2}/\{i,i+1\}},v_{I_{1}},v_{i}',v_{i+1},v_{I_{2}/\{i,i+1\}})
\end{multline*}

\begin{multline*}
=\shoveright{\underset{\begin{array}{c}
x_{I_{2}/\{i+1\}}\\
y_{I_{2}/\{i+1\}}
\end{array}}{\sum}\underset{\begin{array}{c}
x_{i+1}\\
y_{i+1}
\end{array}}{\sum}P_{XY|UV}(x,y|u_{I_{1}},u'_{\{i,i+1\}},u_{I_{2}/\{i,i+1\}},v_{I_{1}},v_{\{i,i+1\}}',v_{I_{2}/\{i,i+1\}})}\\
=...=\underset{x_{I_{2}},y_{I_{2}}}{\sum}P_{XY|UV}(x,y|u_{I_{1}},u'_{I_{2}},v_{I_{1}},v'_{I_{2}}).\tag*{\qedhere}
\end{multline*}

\end{proof}
Combining Lemma \ref{lem:our-imply-sequentialy-signaling} together
with Theorem \ref{thm:main} implies the following.
\begin{cor}
\label{cor:seq-cor}There exists an almost backward non-signalling
system as in Definition \ref{sequential-signaling} such that for
any hash function $f$, there exists a partition $w$ for which the
distance from uniform of $f(X)$ given $w$ is at least $c(\varepsilon)$,
i.e., $d(f(X)|Z(w))\geq c(\varepsilon)$, where $c(\varepsilon)$
is some constant which depends only on the error of a single box,
$\varepsilon$ (as in Definition \ref{A-product-system}).
\end{cor}
Another interesting example is the set of equations which includes
non-signalling conditions between all of Alice's systems alone and
non-signalling conditions between all of Bob's systems alone, together
with the condition of non-signalling between Alice and Bob.
\begin{defn}
\label{completly-Alice-completly-Bob}An n-party conditional probability
distribution $P_{XY|UV}$ over $X,Y,U,V\in\{0,1\}^{n}$ is completely
non-signalling on Alice's side and completely non-signalling on Bob's
side, if for any $i\in[n]$, 
\begin{eqnarray*}
\forall x_{\overline{i}},u_{i},u'_{i},u_{\overline{i}}\quad\underset{x_{i}}{\sum}P_{X|U}(x_{i},x_{\overline{i}}|u_{i},u_{\overline{i}}) & = & \underset{x_{i}}{\sum}P_{X|U}(x_{i},x_{\overline{i}}|u'_{i},u_{\overline{i}})\\
\forall y_{\overline{i}},v_{i},v'_{i},v_{\overline{i}}\quad\underset{y_{i}}{\sum}P_{Y|V}(y_{i},y_{\overline{i}}|v_{i},v_{\overline{i}}) & = & \underset{y_{i}}{\sum}P_{Y|V}(y_{i},y_{\overline{i}}|v'_{i},v_{\overline{i}})
\end{eqnarray*}
where $P_{X|U}$ is the marginal system of $P_{XY|UV}$, held by Alice,
and $P_{Y|V}$ is the marginal system of $P_{XY|UV}$, held by Bob.\end{defn}
\begin{lem}
\label{lem:our-imply-completly-Alice}The non-signalling conditions
of Definition \ref{completly-Alice-completly-Bob} are implied by
the non-signalling conditions of Definition \ref{our-system}. \end{lem}
\begin{proof}
We show that this is true for Alice's side. The proof for Bob's side
is analogous. First, for any $i\in[n]$, we can write the equation
\[
\forall x_{\overline{i}},u_{i},u'_{i},u_{\overline{i}}\qquad\underset{x_{i}}{\sum}P_{X|U}(x_{i},x_{\overline{i}}|u_{i},u_{\overline{i}})=\underset{x_{i}}{\sum}P_{X|U}(x_{i},x_{\overline{i}}|u'_{i},u_{\overline{i}})
\]
using the original system $P_{XY|UV}$ and the definition of a marginal
system:
\[
\forall x_{\overline{i}},u_{i},u'_{i},u_{\overline{i}},v\qquad\underset{\begin{array}{c}
x_{i},y\end{array}}{\sum}P_{XY|UV}(x,y|u_{i},u_{\overline{i}},v)=\underset{\begin{array}{c}
x_{i},y\end{array}}{\sum}P_{XY|UV}(x,y|u'_{i},u_{\overline{i}},v).
\]

Now, as in the proof of Lemma \ref{lem:our-imply-sequentialy-signaling},
\begin{align*}
\underset{\begin{array}{c}
x_{i},y\end{array}}{\sum}P_{XY|UV}(x,y|u_{i},u_{\overline{i}},v) & =\underset{y/\{y_{i}\}}{\sum}\underset{\begin{array}{c}
x_{i},y_{i}\end{array}}{\sum}P_{XY|UV}(x,y|u_{i},u_{\overline{i}},v)\\
 & =\underset{y/\{y_{i}\}}{\sum}\underset{\begin{array}{c}
x_{i},y_{i}\end{array}}{\sum}P_{XY|UV}(x,y|u'_{i},u_{\overline{i}},v)\\
 & =\underset{\begin{array}{c}
x_{i},y\end{array}}{\sum}P_{XY|UV}(x,y|u'_{i},u_{\overline{i}},v).\tag*{\qedhere}
\end{align*}

\end{proof}
Combining Lemma \ref{lem:our-imply-completly-Alice} together with
Theorem \ref{thm:main} implies the following.
\begin{cor}
There exists a system as in Definition \ref{completly-Alice-completly-Bob}
such that for any hash function $f$, there exists a partition $w$
for which the distance from uniform of $f(X)$ given $w$ is at least
$c(\varepsilon)$, i.e., $d(f(X)|Z(w))\geq c(\varepsilon)$, where
$c(\varepsilon)$ is some constant which depends only on the error
of a single box, $\varepsilon$ (as in Definition \ref{A-product-system}).
\end{cor}

\section{Privacy Amplification Against Non-signalling Adversaries \label{sec:Privacy-Amplification-Against}}

\subsection{\label{sub:The-impossibility-without-n.s.}The impossibility of privacy
amplification under the basic non-signalling assumptions }

We use here the same adversarial strategy as presented in \cite{hanggi2010impossibility}
and therefore repeat it here shortly for completeness. For additional
intuitive explanations and complete formal proofs please see \cite{hanggi2010impossibility}.

As explained before, Alice's and Bob's goal is to create a highly
secure key using a system, $P_{XY|UV}$, shared by both of them. Eve's
goal is to get some information about the key. It is therefore natural
to model this situation in the following way: Alice, Bob and Eve share
together a system $P_{XYZ|UVW}$, an extension of the system $P_{XY|UV}$
held by Alice and Bob, which fulfills some known non-signalling conditions.
Each party can perform measurements on its part of the system (i.e.,
insert input and read the outputs of their interfaces of the system),
communicate using a public authenticated channel, Alice then applies
some public hash function $f$ to the outcome she holds, $X$, and
in the end Alice should have a key $K=f(X)$, which is $\epsilon$-indistinguishable
from an ideal, uniformly distributed key, even conditioned on Eve's
information. I.e., $d(K|Z(W))\leq\epsilon$. 

The distance from uniform of the key $k$ is lower-bounded by the
distance from uniform of a single bit of the key, and therefore, for
an impossibility result, it is enough to assume that $f$ outputs
just one bit. Note that since the adversarial strategy can be chosen
after all public communication is over, it can also depend on a random
seed for the hash function. Therefore it is enough to consider deterministic
functions in this case. 

We consider a partition with only two outputs, $z=0$ and $z=1$,
each occurring with probability $\frac{1}{2}$, such that given $z=0$,
$f(X)$ is maximally biased towards 0. According to Lemma \ref{lem:one-part-enough}
it is enough to explicitly construct the conditional system given
measurement outcome $z=0$. In order to do so we start from the unbiased
system as seen by Alice and Bob and {}``shift around'' probabilities
such that $f(X)$ is maximally biased towards 0 and the marginal system
remains valid. By valid me mean that:
\begin{enumerate}
\item All entries must remain probabilities between 0 and 1.
\item The normalization of the probability distribution must remain.
\item The non-signalling condition between Alice and Bob must be satisfied.
\item There must exist a second measurement outcome $z=1$ occurring with
probability $\frac{1}{2}$, and such that the conditional system,
given outcome $z=1$, is also a valid probability distribution. This
second system must be able to compensate for the shifts in probabilities.
According to Lemma \ref{lem:one-part-enough} this means that the
entry in every cell must be smaller or equal twice the original entry.
\end{enumerate}
The system $P_{XY|UV}^{z=0}$ which describes this strategy is defined
formally in the following way. For simplicity we will drop the subscript
of $P_{XY|UV}(x,y|u,v)$ and write only $P(x,y|u,v)$. We use the
same notations as in \cite{hanggi2010impossibility,hanggi2010device}
and define the following groups: 
\begin{eqnarray*}
y_{<} & = & \left\{ y\biggm|\sum_{x|f(x)=0}P(x,y|u,v)<\sum_{x|f(x)=1}P(x,y|u,v)\right\} \\
y_{>} & = & \left\{ y\biggm|\sum_{x|f(x)=0}P(x,y|u,v)>\sum_{x|f(x)=1}P(x,y|u,v)\right\} \\
x_{0} & = & \left\{ x\biggm|f(x)=0\right\} \\
x_{1} & = & \left\{ x\biggm|f(x)=1\right\} 
\end{eqnarray*}
and a factor $c(x,y|u,v)$ as:
\begin{eqnarray*}
\forall x & \in & x_{0},y\in y_{<}\qquad c(x,y|u,v)=2\\
\forall x & \in & x_{1},y\in y_{<}\qquad c(x,y|u,v)=\frac{\underset{x'}{\sum}(-1)^{\left(f(x')+1\right)}P(x',y|u,v)}{\underset{x'|f(x')=1}{\sum}P(x',y|u,v)}\\
\forall x & \in & x_{0},y\in y_{>}\qquad c(x,y|u,v)=\frac{\underset{x'}{\sum}P(x',y|u,v)}{\underset{x'|f(x')=0}{\sum}P(x',y|u,v)}\\
\forall x & \in & x_{1},y\in y_{>}\qquad c(x,y|u,v)=0
\end{eqnarray*}

The system $P^{z=0}$ is then defined as $P^{z=0}(x,y|u,v)=c(x,y|u,v)\cdot P(x,y|u,v)$.

Intuitively, this definition of the strategy means that for each $u,v$
and within each row, Eve shifts as much probability as possible out
from the cells $P(x,y|u,v)$ for which $f(x)=1$ and into the cells
$P(x',y|u,v)$ for which $f(x')=0$ (she wants $P^{z=0}$ to be biased
towards 0). The factor $c(x,y|u,v)$ is defined in such a way that
as much probability as possible is being shifted, while still keeping
the system $P^{z=0}$ a valid element of a partition. 

Although Eve shifts probabilities for each $u,v$ separately, $P^{z=0}$
will still fulfill the required non-signalling conditions, which connect
the inputs $u,v$ to other inputs $u',v'$; this is due to the high
symmetry in the original marginal box of Alice and Bob (Definition
\ref{A-product-system}). For example, it is easy to see that since
Eve only shifts probabilities within the same row (i.e. cells with
the same value of $y$) Bob cannot signal to Alice using $P^{z=0}$;
the sum of the probabilities in one row stays the same as it was in
$P$, and since $P$ did not allow for signalling from Bob to Alice,
so do $P^{z=0}$. The other non-signalling conditions follow from
a bit more complex symmetries. 

It was proven in \cite{hanggi2010impossibility} that for this strategy%
\footnote{Actually, this strategy is being used only when Alice is using an
hash function which does not allow Bob to generate a bit from his
output of the system $Y$, which is highly correlated with the key.
If Alice uses a function which does allow Bob to get an highly correlated
key, then this function has to be biased and therefore Eve can just
use the trivial strategy of doing nothing. For more details please
see \cite{hanggi2010impossibility}.%
} $d(K|Z(w))\leq\frac{-1+\sqrt{1+64\varepsilon^{2}}}{32\varepsilon}$.

\subsection{Proof of the theorem - a more general impossibility result}

In order to prove Theorem \ref{thm:main} we will just prove that
the adversarial strategy presented in \cite{hanggi2010impossibility}
still works. Formally, this means that we need to prove that the element
$\left(p^{z=0}=\frac{1}{2},P^{z=0}(x,y|u,v)\right)$ in the partition
is still valid, even when we add the assumptions of Definition \ref{our-system},
and that $d(K|Z(w))$ is high. Since we do not change the strategy,
the same bound on $d(K|Z(w))$ still holds. Moreover, it was already
proven in \cite{hanggi2010impossibility} that $P^{z=0}(x,y|u,v)$
does not allow signalling between Alice and Bob, therefore we only
need to prove that our additional non-signalling assumptions of Definition
\ref{our-system} hold in the system $P^{z=0}(x,y|u,v)$, i.e., the
system satisfies our assumptions even conditioned on Eve's result. 

\begin{onehalfspace}
The first three lemmas deal with the impossibility of signalling from
Alice's side and the next three lemmas deal with Bob's side. All the
lemmas use the high symmetry of the marginal box (Definition \ref{A-product-system}).
What these lemmas show is that most of this symmetry still exists
in $P^{z=0}$, because we only shift probabilities within the same
row.
\end{onehalfspace}

We use the following notation; for all $i\in[n]$ let $u^{i'}$ be
$u^{i'}=u_{1}...u_{i-1},\overline{u_{i}},u_{i+1}...u_{n}$ (i.e.,
only the i'th bit is flipped) and the same for $x^{i'}$, $y^{i'}$
and $v^{i'}$.
\begin{lem}
\begin{onehalfspace}
\label{lem:Alice-equiv-cell}For all $i\in[n]$ and for all $x$,$y$,$u$,$v$
such that $v_{i}=1$, $P(x,y^{i'}|u,v)=P(x,y|u^{i'},v)$.\end{onehalfspace}
\end{lem}
\begin{proof}
\begin{onehalfspace}
For every single box , $P_{X_{i}Y_{i}|U_{i}V_{i}}(x_{i},y_{i}|u_{i},v_{i})=P_{X_{i}Y_{i}|U_{i}V_{i}}(\overline{x_{i}},\overline{y_{i}}|u_{i},v_{i})$.
Therefore it also holds that $P(x,y|u,v)=P(x^{i'},y^{i'}|u,v)$. Moreover,
\begin{eqnarray*}
P(x,y|u^{i'},v) & = & \left(\frac{1}{2}-\frac{\varepsilon}{2}\right)^{\underset{l}{\sum}1\oplus x_{l}\oplus y_{l}\oplus u_{l}^{i'}\cdot v_{l}}\cdot\left(\frac{\varepsilon}{2}\right)^{\underset{l}{\sum}x_{l}\oplus y_{l}\oplus u_{l}^{i'}\cdot v_{l}}=\\
 & = & \left(\frac{1}{2}-\frac{\varepsilon}{2}\right)^{\underset{l}{\sum}1\oplus x_{l}^{i'}\oplus y_{l}\oplus u_{l}\cdot v_{l}}\cdot\left(\frac{\varepsilon}{2}\right)^{\underset{l}{\sum}x_{l}^{i'}\oplus y_{l}\oplus u_{l}\cdot v_{l}}=\\
 & = & P(x^{i'},y|u,v)
\end{eqnarray*}
Combining these two properties together, we get $P(x,y|u^{i'},v)=P(x^{i'},y|u,v)=P(x,y^{i'}|u,v)$.\end{onehalfspace}
\end{proof}
\begin{lem}
\begin{onehalfspace}
\label{lem:Alice-equiv-type}For all $i\in[n]$ and for all $x$,$y$,$u$,$v$
such that $v_{i}=1$, $c(x,y^{i'}|u,v)=c(x,y|u^{i'},v)$. I.e., the
cells $P(x,y^{i'}|u,v)$ and $P(x,y|u^{i'},v)$ are from the same
type ($x_{0}/x_{1},\: y_{>}/y_{<}$).\end{onehalfspace}
\end{lem}
\begin{proof}
First, it is clear that if $P(x,y^{i'}|u,v)$ was a $x_{0}$ ($x_{1}$)
cell, so is $P(x,y|u^{i'},v)$ because this only depends on $x$. 

Now note that Lemma \ref{lem:Alice-equiv-cell} is correct for every
$x$, therefore the entire row $P(\bullet,y^{i'}|u,v)$ is equivalent
to the row $P(\bullet,y|u^{i'},v)$. This means that if we change
$y^{i'}$ to $y$ and $u$ to $u^{i'}$ together, we will get the
same row, and therefore if $P(x,y^{i'}|u,v)$ was a $y_{<}$ ($y_{>}$)
cell, so is $P(x,y|u^{i'},v)$. All together we get $c(x,y^{i'}|u,v)=c(x,y|u^{i'},v)$.
\end{proof}
The properties of the marginal system $P_{XY|UV}$ which are being
used in Lemma \ref{lem:Alice-equiv-cell} and Lemma \ref{lem:Alice-equiv-type}
can be easily seen, for example, in Table \ref{tab:11} and Table
\ref{tab:10}. For simplicity we consider a product of only 2 systems.
When changing Alice's input from $u=11$ to $u=10$ while $v=11$,
the rows interchange as Lemma \ref{lem:Alice-equiv-cell} suggests.

\begin{table}[b]
\centering  
\begin{tabular}{c c | c c c c} 
	& &\multicolumn{4}{c}{$u=11$} \\
	&\backslashbox{y}{x} & 00 & 01 & 10 & 11 \\
	\hline\\[-1.1em]
	\multirow{4}{*}{$v=11$} & 00 & 
								$(\frac{\varepsilon}{2})^2$ \cellcolor[gray]{0.9} &
								$\frac{\varepsilon}{2}\cdot\frac{1-\varepsilon}{2}$ \cellcolor[gray]{0.9} &
								$\frac{\varepsilon}{2}\cdot\frac{1-\varepsilon}{2}$ \cellcolor[gray]{0.9} &
								$(\frac{1-\varepsilon}{2})^2$ \cellcolor[gray]{0.9} \\[1ex] 
							& 01 &
								$\frac{\varepsilon}{2}\cdot\frac{1-\varepsilon}{2}$ & 
								$(\frac{\varepsilon}{2})^2$ & 
								$(\frac{1-\varepsilon}{2})^2$ & 
								$\frac{\varepsilon}{2}\cdot\frac{1-\varepsilon}{2}$ \\[1ex] 
							& 10 & 
								$\frac{\varepsilon}{2}\cdot\frac{1-\varepsilon}{2}$ \cellcolor[gray]{0.9} &
								$(\frac{1-\varepsilon}{2})^2$ \cellcolor[gray]{0.9} &
								$(\frac{\varepsilon}{2})^2$ \cellcolor[gray]{0.9} &
								$\frac{\varepsilon}{2}\cdot\frac{1-\varepsilon}{2}$ \cellcolor[gray]{0.9} \\[1ex]
							& 11 &
								$(\frac{1-\varepsilon}{2})^2$ &
								$\frac{\varepsilon}{2}\cdot\frac{1-\varepsilon}{2}$ &
								$\frac{\varepsilon}{2}\cdot\frac{1-\varepsilon}{2}$ &
								$(\frac{\varepsilon}{2})^2$ \\[1ex]
	\hline 
\end{tabular}

\caption{\label{tab:11}$P_{XY|UV}$ for two systems ($n=2$), for $u=11$,
$v=11$}
\end{table}

\begin{table}
\centering  
\begin{tabular}{c c | c c c c} 
	& &\multicolumn{4}{c}{$u=10$} \\
	&\backslashbox{y}{x} & 00 & 01 & 10 & 11 \\
	\hline\\[-1.0em]
	\multirow{4}{*}{$v=11$} & 00 & 
								$\frac{\varepsilon}{2}\cdot\frac{1-\varepsilon}{2}$ & 
								$(\frac{\varepsilon}{2})^2$ & 
								$(\frac{1-\varepsilon}{2})^2$ & 
								$\frac{\varepsilon}{2}\cdot\frac{1-\varepsilon}{2}$ \\[1ex]
							& 01 &
								$(\frac{\varepsilon}{2})^2$ \cellcolor[gray]{0.9} &
								$\frac{\varepsilon}{2}\cdot\frac{1-\varepsilon}{2}$ \cellcolor[gray]{0.9} &
								$\frac{\varepsilon}{2}\cdot\frac{1-\varepsilon}{2}$ \cellcolor[gray]{0.9} &
								$(\frac{1-\varepsilon}{2})^2$ \cellcolor[gray]{0.9} \\[1ex] 
							& 10 & 
								$(\frac{1-\varepsilon}{2})^2$ &
								$\frac{\varepsilon}{2}\cdot\frac{1-\varepsilon}{2}$ &
								$\frac{\varepsilon}{2}\cdot\frac{1-\varepsilon}{2}$ &
								$(\frac{\varepsilon}{2})^2$ \\[1ex]
							& 11 &
								$\frac{\varepsilon}{2}\cdot\frac{1-\varepsilon}{2}$ \cellcolor[gray]{0.9} &
								$(\frac{1-\varepsilon}{2})^2$ \cellcolor[gray]{0.9} &
								$(\frac{\varepsilon}{2})^2$ \cellcolor[gray]{0.9} &
								$\frac{\varepsilon}{2}\cdot\frac{1-\varepsilon}{2}$ \cellcolor[gray]{0.9} \\[1ex]
	\hline 
\end{tabular}

\caption{\label{tab:10}$P_{XY|UV}$ for two systems ($n=2$), for $u=10$,
$v=11$}
\end{table}

\begin{lem}
\begin{onehalfspace}
\label{lem:proof-Alice-side}In the conditional system \textup{$P^{z=0}$}
, for any $i\in[n]$ 
\[
\forall x_{\overline{i}},y_{\overline{i}},u_{i},u_{\overline{i}},v\qquad\underset{x_{i},y_{i}}{\sum}P^{z=0}(x,y|u,v)=\underset{x_{i},y_{i}}{\sum}P^{z=0}(x,y|u^{i'},v).
\]
\end{onehalfspace}
\end{lem}
\begin{proof}
\begin{onehalfspace}
First note that for any $u$ and $v$ such that $v_{i}=0$ the probability
distribution $P_{XY|U=u,V=v}$ is identical to $P_{XY|U=u^{i'},V=v}$
(because of the properties of a single box, see Figure \ref{fig:PR-box-error}).
Therefore Eve will shift the probabilities in these two systems in
the same way, which implies that $P_{XY|U=u,V=v}^{z=0}$ is identical
to $P_{XY|U=u^{i'},V=v}^{z=0}$, and in particular, any non-signalling
conditions will hold in this case. 

Assume $v_{i}=1$. We will prove something a bit stronger than needed.
We prove that for all $x,y_{\overline{i}},u_{i},u_{\overline{i}},v$,
$\underset{y_{i}}{\sum}P^{z=0}(x,y|u,v)=\underset{y_{i}}{\sum}P^{z=0}(x,y|u^{i'},v)$.
This in particular implies that $\underset{x_{i},y_{i}}{\sum}P^{z=0}(x,y|u,v)=\underset{x_{i},y_{i}}{\sum}P^{z=0}(x,y|u^{i'},v)$
also holds. 
\begin{align*}
\underset{y_{i}}{\sum}P^{z=0}(x,y|u^{i'},v) & =\underset{y_{i}}{\sum}c(x,y|u^{i'},v)\cdot P(x,y|u^{i'},v)\\
 & =\underset{y_{i}}{\sum}c(x,y^{i'}|u,v)\cdot P(x,y^{i'}|u,v)\\
 & =\underset{y_{i}}{\sum}P^{z=0}(x,y^{i'}|u,v)\\
 & =\underset{y_{i}}{\sum}P^{z=0}(x,y|u,v).
\end{align*}
The first and third equalities are by the definition of $P^{z=0}$,
the second equality is due to Lemma \ref{lem:Alice-equiv-cell} and
Lemma \ref{lem:Alice-equiv-type} and the last equality is due the
fact that the sum is over $y_{i}$. \end{onehalfspace}
\end{proof}
\begin{lem}
\begin{onehalfspace}
\label{lem:Bob-equiv-cell}For all $i\in[n]$ and for all $x$,$y$,$u$,$v$,
$P(x,y^{i'}|u,v)=P(x,y|u,v^{i'})$.\end{onehalfspace}
\end{lem}
\begin{proof}
\begin{onehalfspace}
\begin{align*}
P(x,y|u,v^{i'}) & =\left(\frac{1}{2}-\frac{\varepsilon}{2}\right)^{\underset{l}{\sum}1\oplus x_{l}\oplus y_{l}\oplus u_{l}\cdot v_{l}^{i'}}\cdot\left(\frac{\varepsilon}{2}\right)^{\underset{l}{\sum}x_{l}\oplus y_{l}\oplus u_{l}\cdot v_{l}^{i'}}=\\
 & =\left(\frac{1}{2}-\frac{\varepsilon}{2}\right)^{\underset{l}{\sum}1\oplus x_{l}\oplus y_{l}^{i'}\oplus u_{l}\cdot v_{l}}\cdot\left(\frac{\varepsilon}{2}\right)^{\underset{l}{\sum}x_{l}\oplus y_{l}^{i'}\oplus u_{l}\cdot v_{l}}=\\
 & =P(x,y^{i'}|u,v).\tag*{\qedhere}
\end{align*}
\end{onehalfspace}
\end{proof}
\begin{lem}
\begin{onehalfspace}
\label{lem:Bob-equiv-type}For all $i\in[n]$ and for all $x$,$y$,$u$,$v$
such that $v_{i}=1$, $c(x,y^{i'}|u,v)=c(x,y|u,v^{i'})$. I.e., the
cells $P(x,y^{i'}|u,v)$ and $P(x,y|u,v^{i'})$ are from the same
type ($x_{0}/x_{1},\: y_{>}/y_{<}$).\end{onehalfspace}
\end{lem}
\begin{proof}
As in Lemma \ref{lem:Alice-equiv-type}, it is clear that if $P(x,y^{i'}|u,v)$
was a $x_{0}$ ($x_{1}$) cell, so is $P(x,y|u,v^{i'})$ because this
only depends on $x$. 

Lemma \ref{lem:Bob-equiv-cell} is correct for every $x$, therefore
the entire row $P(\bullet,y^{i'}|u,v)$ is equivalent to the row $P(\bullet,y|u,v^{i'})$
and therefore if $P(x,y^{i'}|u,v)$ was a $y_{<}$ ($y_{>}$) cell,
so is $P(x,y|u,v^{i'})$. All together we get $c(x,y^{i'}|u,v)=c(x,y|u,v^{i'})$.\end{proof}
\begin{lem}
\begin{onehalfspace}
\label{lem:proof-Bob-side}In the conditional system \textup{$P^{z=0}$}
, for any $i\in[n]$ 
\[
\forall x_{\overline{i}},y_{\overline{i}},u,v_{i},v_{\overline{i}}\qquad\underset{x_{i},y_{i}}{\sum}P^{z=0}(x,y|u,v)=\underset{x_{i},y_{i}}{\sum}P^{z=0}(x,y|u,v^{i'}).
\]
\end{onehalfspace}
\end{lem}
\begin{proof}
\begin{onehalfspace}
In an analogous way to the proof of Lemma \ref{lem:proof-Alice-side},
if $u_{i}=0$ the proof is trivial. Assume $u_{i}=1$.

We prove that for all $x,y_{\overline{i}},u,v_{i},v_{\overline{i}}$,
$\underset{y_{i}}{\sum}P^{z=0}(x,y|u,v)=\underset{y_{i}}{\sum}P^{z=0}(x,y|u,v^{i'})$.
This in particular implies that $\underset{x_{i},y_{i}}{\sum}P^{z=0}(x,y|u,v)=\underset{x_{i},y_{i}}{\sum}P^{z=0}(x,y|u,v^{i'})$
also holds.

\begin{align*}
\underset{y_{i}}{\sum}P^{z=0}(x,y|u,v^{i'}) & =\underset{y_{i}}{\sum}c(x,y|u,v^{i'})\cdot P(x,y|u,v^{i'})=\\
 & =\underset{y_{i}}{\sum}c(x,y^{i'}|u,v)\cdot P(x,y^{i'}|u,v)=\\
 & =\underset{y_{i}}{\sum}P^{z=0}(x,y^{i'}|u,v)=\\
 & =\underset{y_{i}}{\sum}P^{z=0}(x,y|u,v).\tag*{\qedhere}
\end{align*}
\end{onehalfspace}

\end{proof}
Note that the only difference between the full non-signalling conditions
and what we have proved here is that in Lemma \ref{lem:proof-Alice-side}
we have to keep the summation over $y_{i}$. Moreover, it is interesting
to see that at least on Bob's side, the {}``full'' non-signalling
conditions also hold in $P^{z=0}$. Since Eve's strategy is defined
to work on each row separately, the symmetry on Bob's side does not
break at all. 

Lemmas \ref{lem:proof-Alice-side} and \ref{lem:proof-Bob-side} together
prove that the assumption of Definition \ref{our-system} holds even
conditioned on Eve's result. Adding this to the rest of the proof
of \cite{hanggi2010impossibility} proves Theorem \ref{thm:main}.

\section{Concluding Remarks and Open Questions \label{sec:Concluding-Remarks}}

In this letter we proved that privacy amplification is impossible
even if we add a lot more non-signalling conditions over the assumptions
of \cite{hanggi2010impossibility}. This also implies that privacy
amplification is impossible under the assumptions of an almost backward
non-signalling system. An interesting question which arises from our
theorem is whether the non-signalling conditions in which the backward
non-signalling systems and the almost backward non-signalling system
differs are the ones which give Eve the tremendous power which makes
privacy amplification impossible. If yes, then it might be the case
that privacy amplification is possible in the relevant setting of
backward non-signalling systems. On the other hand, if the answer
to this question is no, then privacy amplification is also impossible
for backward non-signalling systems. If this is indeed the case then
it seems that the security proof for any practical QKD protocol will
have to be based on quantum physics somehow, and not on the non-signalling
postulate alone. 

Another interesting question is whether we can extend our result to
the case where Alice and Bob use a more interactive protocol to amplify
the secrecy of their key; instead of just applying some hash function
only on Alice's output $X$ and get a key $K=f(X)$, maybe they can
use Bob's output $Y$ as well and create a key $K=g(X,Y)$.

\paragraph*{Acknowledgments:}

Rotem Arnon Friedman thanks Renato Renner for helpful discussions.
Amnon Ta-Shma and Rotem Arnon Friedman acknowledge support from the
FP7 FET-Open project QCS. Esther Hänggi acknowledges support from
the National Research Foundation (Singapore) and the Ministry of Education
(Singapore).

\bibliographystyle{plain}
\bibliography{biblopropos}

\end{document}